\documentclass[letter,11pt,reqno]{article}

\usepackage{relsize}
\usepackage{cite}
\usepackage{color}
\usepackage[dvipsnames]{xcolor}

\usepackage{hyperref, enumitem}
\hypersetup{
  colorlinks   = true, 
  urlcolor     = orange, 
  linkcolor    = Purple, 
  citecolor   = red 
}

\usepackage{amsmath,amsthm,amssymb}
\usepackage[margin=3cm]{geometry}

\usepackage{setspace}
\usepackage{amsmath, amsthm, amssymb, mathrsfs, mathtools}
\usepackage{graphicx}
\usepackage{color}
\usepackage{hyperref}
\usepackage{float}
\usepackage{enumitem}
\usepackage{tikz-network}
\usepackage{xfrac}
\usepackage{todonotes}
\usepackage{subfig}

\theoremstyle{definition}
\newtheorem{theorem}{Theorem}
\newtheorem{lemma}[theorem]{Lemma}
 
\newtheorem{proposition}[theorem]{Proposition}

\theoremstyle{definition}
\newtheorem{definition}[theorem]{Definition}
\newtheorem{example}[theorem]{Example}
\newtheorem{notation}[theorem]{Notation}
\newtheorem{remark}[theorem]{Remark}
\newtheorem*{model}{Problem Formulation}
\newtheorem{claim}{Claim}

\newcommand*{\myproofname}{Proof}

\newcommand{\N}{\mathbb{N}}
\newcommand{\R}{\mathbb{R}}
\newcommand{\Q}{\mathbb{Q}}
\newcommand{\F}{\mathbb{F}}

\newcommand{\calN}{\mathcal{N}}

\newcommand{\calL}{\mathcal{L}}
\newcommand{\mN}{\mathcal{N}}
\newcommand{\mE}{\mathcal{E}}
\newcommand{\mV}{\mathcal{V}}

\newcommand{\mF}{\mathcal{F}}
\newcommand{\mL}{\mathcal{L}}

\newcommand{\IS}{\textnormal{IS}}
\newcommand{\calR}{\mathcal{R}}

\newcommand{\gen}[1]{\langle #1\rangle}

\newcommand{\rank}{\operatorname{rank}}

\newcommand{\Fan}{\textnormal{Fan}}

\usepackage{etoolbox}

\title{External Codes for Multiple
Unicast Networks \\ via Interference Alignment\footnote{Data sharing not applicable to this article as no datasets were generated or analysed during the current study.}}

\usepackage{authblk}
\newcommand*\samethanks[1][\value{footnote}]{\footnotemark[#1]}

\author[1]{F. R. Kschischang
\thanks{Supported by a
Discovery Grant from the Natural Sciences and Engineering Research
Council of Canada.}}
\affil[1]{University of Toronto, Canada}

\author[2]{F. Manganiello
\thanks{Partially supported by the National Science Foundation under grant
DMS-1547399.}}
\affil[2]{Clemson University, U.S.A.}

\author[3]{A. Ravagnani
\thanks{Supported
by the Dutch Research Council through grants VI.Vidi.203.045 and OCENW.KLEIN.539, by the Royal Academy of Arts and Sciences of the Netherlands, and by the European Commission via the ENCODE project.}}
\affil[3]{Eindhoven University of Technology, the Netherlands}

\author[2]{K. Savary
\samethanks[2]}

\date{}

\begin{document}

\maketitle
\begin{abstract}
We introduce a formal framework to study the multiple unicast problem for a coded network in which the network code is linear over a finite field and fixed.
We show that the problem corresponds to an interference alignment problem over a finite field.
In this context, we establish an outer bound for the achievable rate region and provide examples of networks where the bound is sharp.
We finally give evidence of the crucial role played by the field characteristic in the problem.
\end{abstract}

\medskip

\section{Introduction}

The typical scenario considered in the context of
\textit{network coding} consists of one or multiple sources of information attempting to communicate 
to multiple terminals through a network of intermediate nodes.
Various communication 
paradigms have been studied 
in this setting, including noisy, adversarial, error-free multicast, and multiple-unicast networks; see~\cite{ACLY00,LYC03,KM03,DFZ05,WGJ11,DVM10,RDMMJV10,RK18} among many others.

To our best knowledge,
in most network coding references 
users are allowed to freely design both the \textit{network code} (i.e., how the intermediate vertices process information packets) and the \textit{external codes} of the sources.
In particular, designing the network code is part of the communication problem.

This paper considers the multiple unicast problem when the network code is linear and \textit{fixed}, and only the external network codes can be freely designed by the source-receiver pairs. 
In this context, source-receiver pairs compete for network resources and interfere with each other.
We argue that,
in this regime,
the multiple unicast problem corresponds to an \textit{interference alignment problem} over finite fields.
While interference alignment over the complex field
has been extensively studied in more classical information theory settings,
methods do not extend to finite fields in any obvious way.
We refer to~\cite{CJ08} and~\cite{ZY16}
for the ``classical'' interference alignment problem.

The concept of interference alignment has also been considered over finite fields, in connection with network coding.
For example, \cite{DVM10} and \cite{RDMMJV10} propose interference alignment solutions, in combination with network coding, to construct schemes for coded, multiple unicast networks.  Related contributions are
\cite{KJ13,KJ14,CJ08,HC14}, all of which study a problem that is different from the one we address in this paper.

This work makes three main contributions: (1) it introduces a  framework for investigating multi-shot interference alignment problems over finite fields;
(2) it establishes an outer bound for the achievable rate regions in the context outlined above and provides examples where the bound is sharp;
(3) it shows how the field characteristic plays a crucial role 
in the solution to this problem.
Note that (3) already played an important role in the networks introduced in \cite{DFZ05}, but is in sharp contrast with what is typically observed in network coding results for unicast networks, where the field \textit{size}, rather than the characteristic, is the main player.

The rest of the paper is organized as follows. 
In Section~\ref{s:matrix-representation} we formally introduce the communication model and formulate the problem on which we focus. In Section~\ref{s:achievable} we formalize the concepts of achievable rate regions, considering the number of channel uses. Section~\ref{s:outer} establishes an outer bound for said achievable rate regions, and Section~\ref{s:char} describes the role played by the field characteristic.
The paper contains several examples that illustrate concepts and results.

\section{Communication Model and Problem Formulation}\label{s:matrix-representation}

Throughout this paper,
$q$ is a prime power and $\F_q$ denotes the finite field with~$q$ elements. We denote the set of positive integers by $\N=\{1,2, \ldots\}$. All vectors in this paper are row vectors. For $\alpha_1,\dots,\alpha_\ell\in \F_q^n$, we denote by $\langle \alpha_1,\dots,\alpha_\ell \rangle$ the $\F_q$-span of $\alpha_1,\dots,\alpha_\ell$. When $\F_p\subseteq \F_q$, we denote by $\langle \alpha_1,\dots,\alpha_\ell \rangle_p$ the $\F_p$-span of $\alpha_1,\dots,\alpha_\ell$.

We start by informally describing the problem studied in this paper, deferring rigorous definitions for later.

\begin{model}
We consider $n$ uncoordinated sources and terminals, denoted by $S_1,...,S_n$ and $T_1,...,T_n$, respectively. Terminal $T_i$ is interested in decoding only the symbols emitted by source $S_i$ (\textit{multiple unicast} problem).
The sources are connected to the terminals via a network of intermediate nodes, $\calN$ (a directed, acyclic multigraph).
The alphabet of the network is~$\F_q$ and each edge has a capacity of 
one symbol. Alphabet symbols combine linearly at the intermediate nodes of the network, i.e., \textit{linear network coding} is used; see~\cite{LYC03,KM03}. 

We are interested in describing the region of achievable rates for a network of the type we just described, assuming that the operations performed by the intermediate nodes are linear and \textit{fixed}. In other words,
sources and terminals cannot change how the network's nodes linearly combine symbols, but are free to agree on a codebook. Under these assumptions, source-receiver pairs compete for the network's resources.
\end{model}

The following example illustrates the problem at hand. 

\begin{example} \label{ex:1}
Figure~\ref{fig:next1} depicts a network with two  sources and terminals. 
Terminal $T_i$ is interested only in decoding messages from source $S_i$. The operations performed by the gray vertices are fixed. We will 
return to this network in Example~\ref{ex:4}.
\end{example}

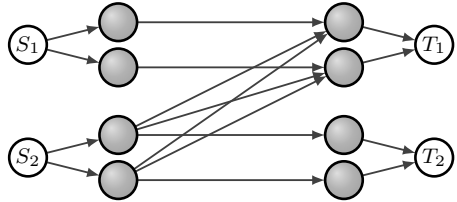
\begin{figure}[H]
\centering
\begin{tikzpicture}[scale=0.6]
\Vertex[label=$S_1$, x=-2, y=0, size=0.5, color=gray, opacity=0]{S1}
\Vertex[label=$S_2$, x=-2, y=-2.5, size=0.5, color=gray, opacity=0]{S2}
\Vertex[label=$T_1$, x=7, y=0, size=0.5, color=gray, opacity=0]{T1}
\Vertex[label=$T_2$, x=7, y=-2.5, size=0.5, color=gray, opacity=0]{T2}
\Vertex[y = 0.5, size=0.5, size=0.5, color=gray, opacity=0.3,style={shading=ball, ball color=gray}]{1}
\Vertex[y = -0.5, size=0.5, size=0.5, color=gray, opacity=0.3,style={shading=ball, ball color=gray}]{2}
\Vertex[x=0, y=-2, size=0.5, size=0.5,  color=gray,  opacity=0.3, style={shading=ball, ball color=gray}]{3}
\Vertex[x=0, y=-3, size=0.5, size=0.5,  color=gray,  opacity=0.3, style={shading=ball, ball color=gray}]{4}
\Vertex[x=5, y=.5, size=0.5, size=0.5,  color=gray,  opacity=0.3, style={shading=ball, ball color=gray}]{7} 
\Vertex[x=5, y=-.5, size=0.5, size=0.5,  color=gray,  opacity=0.3, style={shading=ball, ball color=gray}]{8} 
\Vertex[x=5, y=-2, size=0.5, size=0.5,  color=gray,  opacity=0.3, style={shading=ball, ball color=gray}]{9}
\Vertex[x=5, y=-3, size=0.5, size=0.5,  color=gray,  opacity=0.3, style={shading=ball, ball color=gray}]{10}
\Edge[Direct, lw=.7](S1)(1)
\Edge[Direct, lw=.7](S1)(2)
\Edge[Direct, lw=.7](S2)(3)
\Edge[Direct, lw=.7](S2)(4)
\Edge[Direct, lw=.7](7)(T1)
\Edge[Direct, lw=.7](8)(T1)
\Edge[Direct, lw=.7](9)(T2)
\Edge[Direct, lw=.7](10)(T2)
\Edge[Direct, lw=.7](3)(7)
\Edge[Direct, lw=.7](4)(7)
\Edge[Direct, lw=.7](3)(8)
\Edge[Direct, lw=.7](4)(8) 
\Edge[Direct, lw=.7](1)(7)
\Edge[Direct, lw=.7](2)(8)
\Edge[Direct, lw=.7](3)(9)
\Edge[Direct, lw=.7](4)(10)
\end{tikzpicture} 
\caption{Network for Example~\ref{ex:1}.}
\label{fig:next1}
\end{figure}

In this paper, we propose the following formal 
definition of a communication network, which will facilitate the analysis of the problem we have described above.

\begin{definition} \label{def:network}
A \textbf{multiple unicast network} is a 4-tuple $$\calN=(\mV,\mE,(S_1, \ldots,S_n), (T_1, \ldots,T_n)),$$ where:
\begin{enumerate}[label=(\Alph*)] \setlength\itemsep{0em}
    \item $(\mV,\mE)$ is a finite, directed, acyclic multigraph, which may include parallel edges;
    \item $n \ge 1$ is an integer;
    \item $S_1, \ldots,S_n \in \mV$ are distinct vertices called \textbf{sources};
    \item $T_1, \ldots,T_n \in \mV$ are distinct vertices 
    called \textbf{terminals} or \textbf{sinks}. 
\end{enumerate}
We also assume that the following hold:
\begin{enumerate}[label=(\Alph*)]  \setlength\itemsep{0em}\setcounter{enumi}{4}
    \item $\{S_1, \ldots, S_n\} \cap \{T_1, \ldots, T_n\}=\emptyset$;
    \item for any $i \in \{1,\ldots,n\}$, 
    there exists a directed path in $(\mV,\mE)$ connecting $S_i$ to $T_i$;
    \item sources do not have incoming edges and terminals do not have outgoing edges;
    \item for every vertex $V \in \mV$, there exists a directed path from $S_i$ to $V$ and from $V$ to $T_j$, for some $i,j \in \{1, \ldots,n\}$.
\end{enumerate}
For $V\in \mV$ we denote by $\partial^-(V)$, respectively  $\partial^+(V)$, the indegree, respectively outdegree, of $V$, meaning the number of edges incident to, respectively from, $V$. Moreover, let $\mV^*=\mV \setminus (\{S_1, \ldots,S_n\} \cup \{T_1, \ldots, T_n\})$ denote the set of nonsource and nonterminal vertices, meaning the set of intermediate network nodes. 
\end{definition}

We assume that 
the intermediate nodes of a network process the alphabet symbols linearly. This is made rigorous by the following concept.

\begin{definition}
Let $\mN$ be as in Definition~\ref{def:network}. A \textbf{linear network code}, or simply \textbf{network code}, for $\mN$ is a 
tuple of matrices
 $$\mF=\left( \mF_V \in \F_q^{\partial^-(V) \times \partial^+(V)}\mid V \in \mV^*\right).$$
\end{definition}

Given a network code, the network operates as follows. Fix an intermediate vertex $V\in\mV^*$ and let 
$a=\partial^-(V)$, 
$b=\partial^+(V)$ be the indegree and outdegree of $V$, respectively.
Then $V$ collects a vector of alphabet symbols $(x_1, \ldots, x_a)$ over the incoming edges, and emits the entries of
$$\begin{pmatrix}
y_1 & \ldots & y_b
\end{pmatrix}=\begin{pmatrix}
x_1 & \ldots & x_a
\end{pmatrix} \cdot \mF_V \in \F_q^b
$$
over the outgoing edges.
This fully specifies how $V$ processes information, provided a linear extension of the partial order of the network edges is fixed. Thus, $x_1$ is the symbol that arrives on the minimum edge (with respect to this linear order) at $V$, $x_2$ is the symbol that arrives on the minimum of the remaining edges, etc., and similarly for $y_1, \dots, y_b$. In this paper, networks are delay-free, and communication is instantaneous.

\begin{example}\label{ex:4}
The same network admits multiple network codes. For this example, consider the network $\calN$ from Example \ref{ex:1} with the intermediate nodes labeled as in Figure \ref{fig:next2} and let us assume that the alphabet is $\F_3$.
 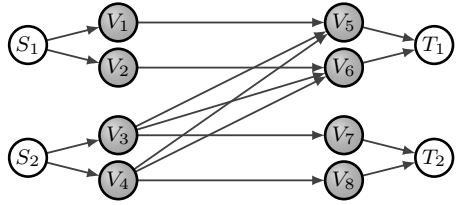
\begin{figure}[H]
\centering
\begin{tikzpicture}[scale=0.6]
\Vertex[label=$S_1$, x=-2, y=0, size=0.5, color=gray, opacity=0]{S1}
\Vertex[label=$S_2$, x=-2, y=-2.5, size=0.5, color=gray, opacity=0]{S2}
\Vertex[label=$T_1$, x=7, y=0, size=0.5, color=gray, opacity=0]{T1}
\Vertex[label=$T_2$, x=7, y=-2.5, size=0.5, color=gray, opacity=0]{T2}
\Vertex[label=$V_1$, y = 0.5, size=0.5, size=0.5, color=gray, opacity=0.3,style={shading=ball, ball color=gray}]{1}
\Vertex[label=$V_2$, y = -0.5, size=0.5, size=0.5, color=gray, opacity=0.3,style={shading=ball, ball color=gray}]{2}
\Vertex[label=$V_3$, x=0, y=-2, size=0.5, size=0.5,  color=gray,  opacity=0.3, style={shading=ball, ball color=gray}]{3}
\Vertex[label=$V_4$, x=0, y=-3, size=0.5, size=0.5,  color=gray,  opacity=0.3, style={shading=ball, ball color=gray}]{4}
\Vertex[label=$V_5$, x=5, y=.5, size=0.5, size=0.5,  color=gray,  opacity=0.3, style={shading=ball, ball color=gray}]{7} 
\Vertex[label=$V_6$, x=5, y=-.5, size=0.5, size=0.5,  color=gray,  opacity=0.3, style={shading=ball, ball color=gray}]{8} 
\Vertex[label=$V_7$, x=5, y=-2, size=0.5, size=0.5,  color=gray,  opacity=0.3, style={shading=ball, ball color=gray}]{9}
\Vertex[label=$V_8$, x=5, y=-3, size=0.5, size=0.5,  color=gray,  opacity=0.3, style={shading=ball, ball color=gray}]{10}
\Edge[Direct, lw=.7](S1)(1)
\Edge[Direct, lw=.7](S1)(2)
\Edge[Direct, lw=.7](S2)(3)
\Edge[Direct, lw=.7](S2)(4)
\Edge[Direct, lw=.7](7)(T1)
\Edge[Direct, lw=.7](8)(T1)
\Edge[Direct, lw=.7](9)(T2)
\Edge[Direct, lw=.7](10)(T2)
\Edge[Direct, lw=.7](3)(7)
\Edge[Direct, lw=.7](4)(7)
\Edge[Direct, lw=.7](3)(8)
\Edge[Direct, lw=.7](4)(8) 
\Edge[Direct, lw=.7](1)(7)
\Edge[Direct, lw=.7](2)(8)
\Edge[Direct, lw=.7](3)(9)
\Edge[Direct, lw=.7](4)(10)
\end{tikzpicture} 
\caption{Network from Example~\ref{ex:1} with labeled vertices.}
\label{fig:next2}
\end{figure} 
  The following are two network codes defined on the network of Figure \ref{fig:next2} where the order of the entries in the tuple follows the order of the intermediate nodes, and the order of the entries in the matrices agrees with the order of the outgoing and incoming channels in the nodes respectively.
\[\mF_1=\left(\begin{pmatrix}1\end{pmatrix},\begin{pmatrix}1\end{pmatrix},\begin{pmatrix}1&1&1\end{pmatrix},\begin{pmatrix}1&1&1\end{pmatrix},\begin{pmatrix}1\\1\\1\end{pmatrix},\begin{pmatrix}1\\1\\1\end{pmatrix},\begin{pmatrix}1\end{pmatrix},\begin{pmatrix}1\end{pmatrix}\right)\]
\[\mF_2=\left(\begin{pmatrix}1\end{pmatrix},\begin{pmatrix}1\end{pmatrix},\begin{pmatrix}1&1&1\end{pmatrix},\begin{pmatrix}1&1&1\end{pmatrix},\begin{pmatrix}1\\1\\1\end{pmatrix},\begin{pmatrix}1\\1\\2\end{pmatrix},\begin{pmatrix}1\end{pmatrix},\begin{pmatrix}1\end{pmatrix}\right)\]
\end{example}

The choice of a linear network code $\mF$ for a
communication network $\mN$ induces end-to-end transfer matrices, one for each source-terminal pair. We denote by $F_{i,j}$ the transfer matrix from $S_i$ to $T_j$.

\begin{example}
Using the network $\calN$ from Example \ref{ex:1}, the network code $\mF_1$ induces the transfer matrices 
\[F_{1,1}=F_{2,2}=\begin{pmatrix}
    1&0\\0&1
\end{pmatrix},\  F_{1,2} = \begin{pmatrix}
    0&0\\0&0
\end{pmatrix},\  \textnormal{and}\ F_{2,1}=\begin{pmatrix}
    1&1\\1&1
\end{pmatrix},\]
whereas the network code $\mF_2$ induces transfer matrices 
\begin{equation*}
    F_{1,1}=F_{2,2}=\begin{pmatrix}
    1&0\\0&1
\end{pmatrix},\  F_{1,2} = \begin{pmatrix}
    0&0\\0&0
\end{pmatrix},\  \textnormal{and}\ F_{2,1}=\begin{pmatrix}
    1&1\\1&2
\end{pmatrix}. \qedhere
\end{equation*}
\end{example}

In one channel use, assuming that $S_i$ transmits $x_i\in \F_q^{\partial^+(S_i)}$, terminal $T_j$ observes the vector $$y_j=\sum_{i=1}^nx_iF_{i,j}\in \F_q^{\partial^-(T_j)}.$$

Since in this paper both the communication network $\mN$ and the linear network code~$\mF$ are supposed to be \textit{fixed},
the end-to-end transfer matrices induced by them fully specify the communication channel.
We, therefore, propose the following definition.

\begin{definition} \label{def:LMUC}
A \textbf{$q$-ary linear multiple unicast channel} (in short, \textbf{$q$-LMUC}) is a 4-tuple $\mL=(n,\pmb{s},\pmb{t},F)$, where $n\in \N$ is a positive integer, $\pmb{s}=(s_1,\dots,s_n),\pmb{t}=(t_1,\dots,t_n)\in \N^n$, and $F\in \F_q^{s\times t}$, where $s=\sum_{i=1}^n s_i$ and $t=\sum_{i=1}^n t_i$. We call $F$ the \textbf{transfer matrix} and regard it as a block matrix
\[F=\begin{pmatrix}F_{1,1}& \cdots & F_{1,n}\\
    \vdots&\ddots&\vdots\\ F_{n,1}&\cdots&F_{n,n}\end{pmatrix},\]
where block $F_{i,j}$ has size $s_i \times t_j$.
\end{definition}

Here $n$ represents the number of source-terminal pairs, $s_i=\partial^+(S_i)$, and $t_i=\partial^-(T_i)$ for $i=1,\dots,n$. Moreover,~$F$ collects the transfer matrices for each source-terminal pair.
Note that we do not need to remember the communication network or the entire network code.
The matrix~$F$ fully describes the end-to-end channel laws for a single channel use. 

We can extend this definition for multiple uses of the channel as follows. Suppose that the network is used $m \ge 1$ times. 
The channel input is an element
$x=(x_1,\dots,x_n)\in \F_{q^m}^{s_1}\times \cdots \times \F_{q^m}^{s_n}=\F_{q^m}^s$. More precisely, for all $i \in \{1,\dots,n\}$, $x_i\in \F_{q^m}^{s_i}$ is the input that source $S_i$ emits. Then the channel output is $\smash{y=(y_1,\dots,y_n)\in \F_{q^m}^{t_1}\times \cdots \times \F_{q^m}^{t_n}=\F_{q^m}^t}$, where  
\begin{equation}
\label{e:model}
    y_i=x_iF_{i,i}+\sum_{j\neq i}x_jF_{j,i}
\end{equation}
is the vector that terminal $T_i$ receives on its incoming edges.
Note that the field extension~$\F_{q^m}$ models $m$ uses of the channel because the network code is assumed to be $\F_q$-linear.

\begin{remark}
Given any $q$-LMUC as in Definition~\ref{def:LMUC},
it is always possible to construct a multiple unicast network $\mN$ and a network code $\mF$ for $\mN$ that induces the given transfer matrix. We illustrate the procedure with an example.
\end{remark}

\begin{example}\label{ex:8}
The $11$-LMUC \[\calL = \left(2,(1,2),(2,2),\begin{pmatrix}
    1 & 0 & 2 & 3\\ 0 & 4 & 5 & 0\\ 6 & 7 & 0 & 0
\end{pmatrix}\right)\]
induces the matrices
\[F_{1,1}=\begin{pmatrix}
    1 & 0
\end{pmatrix},\quad  F_{1,2}=\begin{pmatrix}
    2 &3
\end{pmatrix},\quad F_{2,1}=\begin{pmatrix}
    0 & 4\\ 6 & 7
\end{pmatrix},\quad  F_{2,2}=\begin{pmatrix}
    5&0\\0&0
\end{pmatrix}.\]
A network code
that induces these matrices
is
\begin{equation*}
    \mF=\left(\begin{pmatrix}
    1 & 1 & 1
\end{pmatrix}, \begin{pmatrix}
    1 & 1
\end{pmatrix}, \begin{pmatrix}
    1 & 1
\end{pmatrix}, \begin{pmatrix}
    1 \\ 6
\end{pmatrix}, \begin{pmatrix}
    4 \\ 7
\end{pmatrix},
\begin{pmatrix}
    2 \\ 5
\end{pmatrix}, \begin{pmatrix}
    3
\end{pmatrix}\right),
\end{equation*}
realized by the 
network 
depicted in Figure~\ref{fig:nextt}.
Note that \begin{equation*}
    \mF=\left(\begin{pmatrix}
    1 & 2 & 3
\end{pmatrix}, \begin{pmatrix}
    4 & 5
\end{pmatrix}, \begin{pmatrix}
    6 & 7
\end{pmatrix}, \begin{pmatrix}
    1 \\ 1
\end{pmatrix}, \begin{pmatrix}
    1 \\ 1
\end{pmatrix},
\begin{pmatrix}
    1 \\ 1
\end{pmatrix}, \begin{pmatrix}
    1
\end{pmatrix}\right),
\end{equation*}
also induces the given transfer matrix.
\end{example}

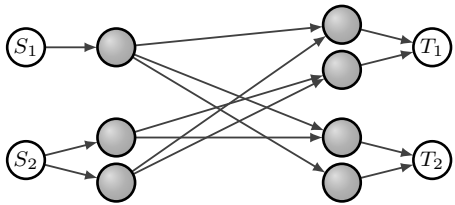
\begin{figure}[H]
\centering
\begin{tikzpicture}[scale=0.6]
\Vertex[label=$S_1$, x=-2, y=0, size=0.5, color=gray, opacity=0]{S1}
\Vertex[label=$S_2$, x=-2, y=-2.5, size=0.5, color=gray, opacity=0]{S2}
\Vertex[label=$T_1$, x=7, y=0, size=0.5, color=gray, opacity=0]{T1}
\Vertex[label=$T_2$, x=7, y=-2.5, size=0.5, color=gray, opacity=0]{T2}
\Vertex[y = 0, size=0.5, size=0.5, color=gray, opacity=0.3,style={shading=ball, ball color=gray}]{1}
\Vertex[x=0, y=-2, size=0.5, size=0.5,  color=gray,  opacity=0.3, style={shading=ball, ball color=gray}]{2}
\Vertex[x=0, y=-3, size=0.5, size=0.5,  color=gray,  opacity=0.3, style={shading=ball, ball color=gray}]{3}
\Vertex[x=5, y=.5, size=0.5, size=0.5,  color=gray,  opacity=0.3, style={shading=ball, ball color=gray}]{7} 
\Vertex[x=5, y=-.5, size=0.5, size=0.5,  color=gray,  opacity=0.3, style={shading=ball, ball color=gray}]{8} 
\Vertex[x=5, y=-2, size=0.5, size=0.5,  color=gray,  opacity=0.3, style={shading=ball, ball color=gray}]{9}
\Vertex[x=5, y=-3, size=0.5, size=0.5,  color=gray,  opacity=0.3, style={shading=ball, ball color=gray}]{10}
\Edge[Direct, lw=.7](S1)(1)
\Edge[Direct, lw=.7](S2)(2)
\Edge[Direct, lw=.7](S2)(3)
\Edge[Direct, lw=.7](7)(T1)
\Edge[Direct, lw=.7](8)(T1)
\Edge[Direct, lw=.7](9)(T2)
\Edge[Direct, lw=.7](10)(T2)
\Edge[Direct, lw=.7](1)(7)
\Edge[Direct, lw=.7](2)(9)
\Edge[Direct, lw=.7](1)(9)
\Edge[Direct, lw=.7](1)(10)
\Edge[Direct, lw=.7](2)(8)
\Edge[Direct, lw=.7](3)(7)
\Edge[Direct, lw=.7](3)(8)
\end{tikzpicture} 
\caption{Network for Example~\ref{ex:8}.}
\label{fig:nextt}
\end{figure}

Given a $q$-LMUC  $\mL=(n,\pmb{s},\pmb{t},F)$
representing a $q$-ary linear multiple unicast channel, we are interested in determining
which rates can be achieved by carefully selecting the codebooks of the sources. Recall that each terminal $T_i$ is interested in decoding \textit{only} its corresponding source~$S_i$. Therefore the information emitted by the sources $S_j$ with $j \neq i$ acts as interference for terminal $T_i$.

\section{Achievable Rate Regions and Their Properties}\label{s:achievable}

In this section, we propose a formal framework to describe the 
capacity of a $q$-LMUC, based on the concept of \textit{fan-out set} and \textit{unambiguous code}.
We start with the concept of  code(book).

\begin{definition} \label{def:code}
Let $\mL=(n,\pmb{s},\pmb{t},F)$ be a $q$-LMUC and let $m \ge 1$ be an integer.
An \textbf{$m$-external code}, or simply an \textbf{$m$-code}, for $\mL$ is a Cartesian product
$C=C_1\times\cdots\times C_n$,
where $C_i \subseteq \F_{q^m}^{s_i}$ for all
$i \in \{1,\ldots,n\}$. 
The elements of each $C_i$ are called \textbf{codewords}.
\end{definition}

Note that we do \textit{not} require that $C$ is a linear space.
Recall that the parameter~$m$ in Definition~\ref{def:code} represents the number of channel uses, as already mentioned when introducing Equation~\eqref{e:model}. In the sequel, 
for a 
$q$-LMUC $\mL=(n,\pmb{s},\pmb{t},F)$
and for $i\in \{1,\dots,n\}$, we denote by $\pi_i:\F_{q^m}^t\rightarrow \F_{q^m}^{t_i}$ the projection on the $i$th block of coordinates (recall that 
$\pmb{t}=(t_1, \ldots,t_n)$ and $t=t_1 + \cdots + t_n$).

\begin{definition} \label{fan-out}
Let $\mL=(n,\pmb{s},\pmb{t},F)$
be a $q$-LMUC, $i \in \{1, \ldots, n\}$, $C$ an $m$-code for $\mL$, and  
$x \in C_i$. We denote by
$$\Fan_i(x,C):=\{\pi_i((x_1,\dots,x_{i-1},x,x_{i+1},\dots, x_n)F) \mid x_j\in C_j \mbox{ for all } 
j \neq i\}$$
the \textbf{$i$-th fan-out set} of $x$ with respect to terminal $i$ and the code $C$. The \textbf{$i$-th fan-out set} of $C$ is $\Fan_i(C)=\cup_{x\in C_i}{\Fan_i(x,C)}\subseteq \F_{q^m}^{t_i}$, for all $i \in \{1,\ldots,n\}$. 
\end{definition}

Following the notation of Definition~\ref{fan-out},
$\Fan_i(x,C)$ is the set of possible words that the $i$th terminal can possibly receive when the $i$th source emits $x$ and the other sources emit their own codewords. Fan-out sets relate to the concept of interference as follows.

\begin{definition} \label{def:IS}
Let $\mL$, $m$, and $C$ be as in Definition~\ref{def:code}.
 We define the \textbf{interference set} of~$C$ at terminal $T_i$ as
$$\IS_i(C)=\Fan_i(0,C)=\left\{\sum_{j\neq i}c_jF_{j,i}\mid c_j\in C_j\right\}.$$
\end{definition}

\begin{remark}
Note that using Definition~\ref{def:IS} we can rewrite equation~\eqref{e:model} as
\begin{equation}\label{e:fan_is}
   \Fan_i(x,C)=xF_{i,i} +\IS_i(C) = \{xF_{i,i} + y \mid y \in \IS_i(C)\}.
\end{equation}
for any $m$-code $C$, any $i \in \{1, \ldots, n\}$, and any $x \in C_i$.
Therefore the $i$-th fan-out set of $x \in C_i$ is a \textbf{translate} (or a \textbf{coset}, if each $C_i$ is linear) of the interference set $\IS_i(C)$.
\end{remark}

Communication 
is considered to be successful when each codeword can be uniquely decoded. The following definition models this concept.

\begin{definition} \label{def:unamb}
Let $\mL$, $m$, and $C$ be as in  Definition~\ref{def:code}. We say that $C$ is \textbf{unambiguous} for $\mL$
if for all $i \in \{1,...,n\}$ and for all codewords $x_1,x_2 \in C_i$ with
$x_1 \neq x_2$, we have
$\Fan_i(x_1,C) \cap \Fan_i(x_2,C) = \emptyset$.
\end{definition}

An unambiguous code, as in the previous definition, uniquely defines decoder maps.
More precisely,
the
\textbf{$i$-th decoder} is the map
$D_i:\Fan_i(C)\rightarrow \F_{q^m}^{s_i}$ defined by $D_i(v_i)=x_i$ for all $v_i \in \Fan_i(C)$, where $x_i \in C_i$ is the only element with 
$v_i \in \Fan_i(x_i,C)$.

\begin{remark} \label{rem:subcode}
    Following the notation of
    Definition~\ref{def:unamb},
    if $C$ is an unambiguous $m$-code for~$\mL$
    and $C'$ is an $m$-code for $\mL$ with $C' \subseteq C$, then
    $C'$ is unambiguous as well.
\end{remark}

We are now ready to give a rigorous definition of \textit{the achievable rate region} of a $q$-LMUC~$\mL$.
For convenience, for $m \ge 1$ we define the set $$\log_{q^m}(\N^n)=\{(\log_{q^m}(u_1),...,\log_{q^m}(u_n)) \mid (u_1,...,u_n) \in \N^n\}.$$

 \begin{definition}\label{def:multi_achievability}
 The \textbf{$m$-shot achievable rate region}
 of a $q$-LMUC $\mL$ as in Definition~\ref{def:LMUC} is the set 
\begin{multline*}
    \mathcal{R}_m(\mL)=\{\alpha\in \log_{q^m}(\N^n) \mid \exists \, C=C_1\times\cdots\times C_n \mbox{ unambiguous $m$-code for $\mL$} \\ \mbox{with }\log_{q^m}(|C_i|) = \alpha_i \ \forall \ 1 \le i \le n \} \subseteq \R_{\ge 0}^n.
\end{multline*}
The \textbf{achievable rate region} of $\mL$ is the set
$$\calR(\mL) = \overline{\bigcup_{m \ge 1} \calR_m(\mL)},$$
where the overline indicates the closure operator with respect to the 
Euclidean topology on~$\R^n$. The elements of $\calR(\mL)$ elements are called \textbf{achievable rates}. 
 \end{definition}
 
 The following example illustrates that different LMUCs supported on the same network might have very different achievable rate regions.

 \begin{example}
The two 3-LMUCs induced by Example \ref{ex:4} are \[\calL_1=\left(2,(2,2),(2,2),\begin{pmatrix}
    1 & 0 & 0 & 0\\ 0 & 1 & 0 & 0\\1 & 1 & 1 & 0\\1 & 1 & 0 &1
\end{pmatrix}\right), \ \ \calL_2=\left(2,(2,2),(2,2),\begin{pmatrix}
    1 & 0 & 0 & 0\\ 0 & 1 & 0 & 0\\1 & 1 & 1 & 0\\1 & 2 & 0 &1
\end{pmatrix}\right).\]
Although the two 3-LMUC have the same underlying network, they have different 1-shot achievable rate regions. Indeed, for $\calL_1$, the code $\langle (1,2)\rangle_{\F_3}\times \F_3^2$  is unambiguous, meaning that $(1,2)\in \calR_1(\calL)$. On the other hand,  $(1,2)\notin \calR_1(\calL_2)$, as we now briefly explain. In~$\calL_2$, the matrix  $F_{2,1}$ is invertible. Moreover, 
a rate of the form $(\alpha_1,2)$ can be achieved only by a code of the form $C=C_1 \times \F_3^2$.
Together with the invertibility of $F_{2,1}$, this implies that $\IS_1(C)=\F_3^2$, making $C=\{0\}\times \F_3^2$ the only unambiguous code with $C_2=\F_3^2$.
 \end{example}

 The following result states a quite intuitive property of achievable rate regions.

 \begin{proposition}
 Let $\mL$ be a $q$-LMUC as in Definition~\ref{def:LMUC} and let $m \ge 1$.
 If $(\alpha_1, \ldots, \alpha_n) \in \calR_m(\mL)$ and $\smash{(\beta_1, \ldots,\beta_n) \in \log_{q^m}(\N^n)}$ satisfies $\beta_i \le \alpha_i$ for all $i \in \{1, \ldots, n\}$, then we have 
 $(\beta_1, \ldots, \beta_n) \in \calR_m(\mL)$.
 \end{proposition}
 
 \begin{proof}
 Since $(\alpha_1, \ldots, \alpha_n) \in \calR_m(\mL)$,
 there exists an $m$-code $C=C_1 \times \cdots \times C_n$ unambiguous for~$\mL$
 with $|C_i|=q^{m \alpha_i}$
 for all $i$. Since $(\beta_1, \ldots, \beta_n) \in \log_{q^m}(\N^n)$, for all $i$ there exists $D_i \subseteq C_i$ with $|D_i|=q^{m \beta_i}$.
 By Remark~\ref{rem:subcode}, the $m$-code
 $D=D_1 \times \cdots \times D_n$ is unambiguous for~$\mL$. This establishes the desired result.
 \end{proof}
 
 We now investigate how the various achievable rate regions relate to each other.
 The following two results are 
 inspired by (but do not immediately follow from) the concept of time sharing.

\begin{proposition} \label{mixed}
Let $\mL$ be a $q$-LMUC as in Definition~\ref{def:LMUC} and let $m,m' \ge 1$. We have
$$\calR_{m+m'}(\mL) \supseteq \frac{m\, \calR_m(\mL) + m' \, \calR_{m'}(\mL)}{m+m'}.$$
\end{proposition}

\begin{proof}
Let $C$ and $C'$ be an $m$-code and an $m'$-code, respectively, for $\mL$.
Let
$\{\gamma_1,...,\gamma_m\}$ and 
$\smash{\{\gamma'_1,...,\gamma'_{m'}\}}$ be 
ordered bases of  
$\smash{\F_{q^m}}$ and $\smash{\F_{q^{m'}}}$ over $\F_q$, respectively.
We denote by $\smash{\varphi_m:\F_q^m \to \F_{q^m}}$ the $\F_q$-isomorphism
$(a_1,\ldots,a_m) \mapsto \sum_{i=1}^m a_i \gamma_i$. Define
$\smash{\varphi_{m'}:\F_q^{m'} \to \F_{q^{m'}}}$
analogously. Fix an ordered basis
$\smash{\{\beta_1, \ldots,\beta_{m+m'}\}}$ of $\F_{q^{m+m'}}$ over $\F_q$. We extend these three $\F_q$-linear maps coordinate-wise.
Take
$$D:= \varphi_m^{-1}(C) \times  \varphi_{m'}^{-1}(C').$$
Then the $(m+m')$-code
$\varphi_{m+m'}(D)$ is unambiguous and it satisfies
$$\log_{q^{m+m'}} \left(\varphi_{m+m'}(D) \right) = \frac{m}{m+m'} \log_{q^m}(|C|) + \frac{m'}{m+m'} \log_{q^{m'}}(|C'|).$$
This establishes the desired result.
\end{proof}

The following result shows that $\calR(\mL)$ contains the convex hull of $\calR_1(\mL)$.

\begin{theorem}
Let $\mL$ be a $q$-LMUC as in Definition~\ref{def:LMUC}.
Denote by $\mbox{conv}(\calR_1(\mL))$ the convex hull of $\calR_1(\mL)$, i.e., the set of convex combinations of the points of $\calR_1(\mL)$. Then
$$\calR(\mL) \supseteq \mbox{conv}(\calR_1(\mL)).$$
\end{theorem}
\begin{proof}
Define $Z=|\calR_1(\mL)|$ and let $\calR_1(\mL)=\{\alpha^{(\ell)} \mid 1 \le \ell \le Z\}$.
Let $A:=\{(a_1,...,a_Z) \in [0,1]^Z \mid a_1+ \cdots +a_Z=1\}$, and let $a \in A$.
We will show that for any real number $\varepsilon >0$, there exists $m \ge 1$
and $\beta \in \calR_m(\mL)$ with
$$\left\| \sum_{\ell=1}^N a_\ell\alpha^{(\ell)} -\beta \right\| \le \varepsilon.$$
This will imply the desired theorem by the definition of closure in the Euclidean topology.

Since the result is clear if $\calR_1(\mL)=\{(0,\dots,0)\}$, we shall assume
$\nu:= \max_\ell \left\| \alpha^{(\ell)} \right\|>0$.
For all $\ell \in \{1, \dots, Z\}$, fix an unambiguous 1-code $\smash{C^{(\ell)} = C^{(\ell)}_1 \times \cdots \times C^{(\ell)}_n}$ that achieves rate~$\alpha_\ell$, i.e., such that
\[\log_q \left(C^{(\ell)}_i  \right) =\alpha_i \mbox{ for all $i$.}\]
Since $A \cap \Q^n$ is dense in $A$, there exists a sequence $(s^k)_{k\in \N}\subseteq A \cap \Q^Z$ with the property that $\lim_{k\rightarrow \infty} s^k=a$.
There exists $k_\varepsilon \in \N$ such that
$\left\| s^k-a \right\| \le \varepsilon/(\nu \sqrt{Z})$ for all
$k \ge k_\varepsilon$. Let $s=s^{k_\varepsilon}$ and $\smash{\beta= \sum_{\ell=1}^Z s_\ell\alpha_\ell}$ and observe that
$$\left\| \sum_{\ell=1}^Z a_\ell\alpha_\ell -\beta \right\| \le \nu \sum_{\ell=1}^Z (s_\ell-a_\ell) \le \nu \sum_{\ell=1}^Z |s_\ell-a_\ell| \le \nu \sqrt{Z} \left\| s-a \right\| \le \varepsilon.$$
Therefore, it suffices to show that $\beta \in \calR_m(\mL)$ for some $m \ge 1$. Write $s=(s_1,...,s_Z)=(b_1/c_1,...,b_Z/c_Z)$, where the $b_\ell$'s and the $c_\ell$'s are integers, and let $m=c_1 \cdots c_Z$. Then $\beta$ can be achieved in $m$ rounds by using code~$C^{(\ell)}$ for $m b_\ell/c_\ell$
rounds, $\ell \in \{1,...,Z\}$, in any order.
A more precise 
formulation can be obtained using 
field extension maps as in the proof of
Proposition~\ref{mixed} (we do not go into the details).
\end{proof}

\section{An Outer Bound for the Achievable Rate Region}
\label{s:outer}

In this section, we establish an outer bound for the achievable rate region of a $q$-LMUC.
Our proof technique derives a lower bound for the size of the fan-out sets introduced in Definition~\ref{fan-out}, which we obtain by estimating the ``amount'' of interference that the users cause to each other. The outer bound is stated in Theorem~\ref{upperBound}, which relies on two preliminary results.

\begin{remark}
Let $M \in \F_q^{a \times b}$ be any matrix. Embed $\F_q$ into $\F_{q^m}$, where $m \ge 1$.
It is well-known that the $\F_q$-rank of $M$ is the same as its $\F_{q^m}$-rank. 
\end{remark}

We start with the following simple observation.

\begin{proposition}
Let $\mL=(n,\pmb{s},\pmb{t},F)$ be a $q$-LMUC and let $m \ge 1$ be an integer.
If $n=1$, then 
$\calR_m(\mL)=\{\alpha \in \log_q(\N) \mid 0 \le \alpha \le \rank F\}$.
\end{proposition}
\begin{proof}
Let $C\subseteq \F_{q^m}^{s_1}$ be an $m$-code. 
Then $C$ is unambiguous if and only if 
$F(x)\neq F(y)$ for all $x,y\in C$ with $x\neq y$, i.e., if and only if 
the elements of $C$ belong to distinct equivalent classes of $\F_{q^m}^{r_1}/\ker(F)$.
This shows that if $C$ is unambiguous, then
$|C| \le q^{m \cdot \rank(F)}$. Vice versa, taking one representative for each class of $\smash{\F_{q^m}^{r_1}/\ker(F)}$ produces an unambiguous $m$-code.
\end{proof}

In the remainder of the section, we show how the argument in the previous proposition extends to an arbitrary number of users. We will need the following preliminary result.

\begin{lemma} \label{lll}
Let $V,W$ be linear spaces over $\F_{q^m}$, $m \ge 1$, and let $L:V \to W$ be an $\F_{q^m}$-linear map. For all non-empty sets $A \subseteq V$, we have
$|L(A)| \cdot |\ker(L)| \ge |A|$.
\end{lemma}

\begin{proof}
Define an equivalence relation on $A$ by setting $a \sim b$ if and only if $L(a)=L(b)$, i.e., if and only if 
$a-b \in \ker(L)$. Then
we know from elementary set theory that
$|A/_\sim| = |f(A)|$.
The equivalence class of $a \in A$ is 
$(a+\ker(L)) \cap B$ and, therefore, it has size at most $|\ker(L)|$. Since the equivalence classes partition $A$, we have $|A| \le |A/_\sim| \cdot  |\ker(L)|$. This concludes the proof.
\end{proof}

We are now ready to state the main result of this section, providing an outer bound for the achievable rate region of a $q$-LMUC. 

\begin{notation}
In the sequel, for a non-empty subset $I \subseteq \{1,...,n\}$ and $j \in I$, we denote by~$F_{I,j}$ the submatrix of $F$ formed by the blocks indexed by $(i,j)$, as $i$ ranges over $I$.
\end{notation}

The main result of this section is the following outer bound.

\begin{theorem}\label{upperBound}
Let $\mL=(n,\pmb{s},\pmb{t},F)$ be a $q$-LMUC, and let $m \ge 1$ be an integer.
Let $(\alpha_1,...,\alpha_n) \in \calR_m(\mL)$. Then for all non-empty $I \subseteq \{1,\ldots,n\}$ and for all $j \in I$, we have
$$\sum_{i \in I} \alpha_i \le \rank(F_{I,j}) - \rank (F_{I \setminus\{j\},j})+ \sum_{\substack{k \in I \\ k \neq j}} s_k.$$
Therefore, for all non-empty $I\subseteq \{1,\dots, n\}$, we have 
$$\sum_{i \in I} \alpha_i \le \min_{j \in I}
\left\{ \rank(F_{I,j}) -\rank(F_{I \setminus\{j\},j})+ \sum_{\substack{k \in I \\ k \neq j}} s_k \right\}.$$ 
\end{theorem}

\begin{proof}
The result easily follows the definitions if $|I|=1$. Now fix an index set with $|I| \ge 2$,
 a tuple $(\alpha_1,...,\alpha_n) \in \calR_m(\mL)$, and an unambiguous $m$-code $C=(C_1,\dots,C_n)$ with $\alpha_i=\log_{q^m}(|C_i|)$ for all $i$. For $i \notin I$, we replace $C_i$ with an arbitrary subset of cardinality one. The resulting code is unambiguous by Remark~\ref{rem:subcode}, and it has $$\log_{q^m}(|C_i|)=
 \begin{cases}
 \alpha_i & \mbox{ if $i \in I$,} \\
 0 & \mbox{ if $i \notin I$.}
 \end{cases}$$  
 The remainder of the proof uses the following lower bound, which we will establish later.
\begin{claim}\label{cla}
For all $j \in I$ and $x \in C_j$, we have
 \begin{equation} \label{kk}
     |\Fan_j(x,C)| \ge \frac{\prod_{\substack{k \in I \\ k \neq j}}|C_k|}{|\ker F_{I \setminus\{j\},j}|}.
 \end{equation}
 \end{claim}
 Note that by fixing $j \in I$ and summing the inequality in Claim~\ref{cla} over all $x \in C_j$ one obtains
 \begin{equation} \label{kk2}
     \sum_{x \in C_j} |\Fan_j(x,C)| \ge \frac{\prod_{k \in I}|C_k|}{|\ker F_{I \setminus\{j\},j}|} \quad \mbox{for all $j \in I$.}
 \end{equation}
 Since $C$ is unambiguous, we have
 $$\sum_{x \in C_j} |\Fan_j(x,C)| = \left|  
 \bigcup_{x \in C_j} \Fan_j(x,C)\right| \le q^{m \cdot \rank(F_{I,j})},$$
 where the latter inequality follows from the fact that 
 $\bigcup_{x \in C_j} \Fan_j(x,C)$ is contained in the image of~$F_{I,j}$. 
 Therefore, the inequality in~\eqref{kk2} implies
 \begin{equation} \label{kk3}
     q^{m \cdot \rank(F_{I,j})} \ge \frac{\prod_{k \in I}|C_k|}{|\ker(F_{I \setminus\{j\},j})|} \quad \mbox{for all $j \in I$.}
 \end{equation}
 We also have $$\dim_{\F_{q^m}} \ker F_{I \setminus\{j\},j} = \sum_{\substack{k \in I \\ k \neq j}} s_k - \rank (F_{I \setminus\{j\},j}).$$
Therefore, taking the logarithm with base $q^m$ in~\eqref{kk3} yields
  \begin{equation} \label{kk4}
     \rank(F_{I,j}) \ge \alpha_1 + \ldots +\alpha_n -\sum_{\substack{k \in I \\ k \neq j}} s_k + \rank(F_{I \setminus\{j\},j}) \quad \mbox{for all $j \in I$,}
 \end{equation}
 which is the desired bound.

It remains to show that Claim~\ref{cla} holds. We only prove it for $I=\{1,...,n\}$ and $j=1$, as the proof for all other cases is the same (but more cumbersome notation-wise). Fix an arbitrary $x \in C_1$ and view 
$F_{1,\{1,...,n\}}$ as an $\F_{q^m}$-linear map $\F_{q^m}^{s_1+ \cdots+s_n} \to \F_{q^m}^{t_1}$. Then $\Fan_1(x,C)$ is the image of
$C_2 \times \cdots \times C_n$ under the map
$$f: \F_{q^m}^{s_2 + \cdots +s_n} \to \F_{q^m}^{t_1}, \qquad (x_2,...,x_n) \mapsto (x,x_2,...,x_n) F_{\{1,...,n\},1}.$$
We have $f= F_{11}x +g$ as functions, where 
$$g: \F_{q^m}^{s_2 + \cdots +s_n} \to \F_{q^m}^{t_1}, \qquad
(x_2,...,x_n) \mapsto (x_2,...,x_n) F_{\{2,...,n\},1}.$$
Therefore, the images of $f$ and $g$ have the same cardinality, $|\Fan_1(x,C)|$.
Finally, by applying Lemma~\ref{lll}
to the $\F_{q^m}$-linear function $g$, we obtain
$$|\Fan_1(x,C)| \cdot |\ker(g)| \ge |C_2 \times \cdots \times C_n|,$$
which is the inequality in Claim~\ref{cla}.
\end{proof}

We illustrate Theorem~\ref{upperBound} with two examples.

\begin{example} \label{eee1}
Let $q$ be arbitrary and consider the $q$-LMUC 
$\mL=(2,(2,2), (2,2),F)$,
where
\[F = \begin{pmatrix} 
1 & 0 & 0 & 0\\
0 & 1 & 1 & 0\\
0 & 1 & 1 & 0\\ 
0 & 0 & 0 & 1
\end{pmatrix}. \]
Note that $\mL$ is induced, for example, by the network in Figure~\ref{fig:next}.
\begin{figure}[H]
\centering
\begin{tikzpicture}[scale=0.6]
\Vertex[label=$S_1$, x=-2, y=0, size=0.5, color=gray, opacity=0]{S1}
\Vertex[label=$S_2$, x=-2, y=-2.5, size=0.5, color=gray, opacity=0]{S2}
\Vertex[label=$T_1$, x=7, y=0, size=0.5, color=gray, opacity=0]{T1}
\Vertex[label=$T_2$, x=7, y=-2.5, size=0.5, color=gray, opacity=0]{T2}
\Vertex[y = 0.5, size=0.5, size=0.5, color=gray, opacity=0.3,style={shading=ball, ball color=gray}]{1}
\Vertex[y = -0.5, size=0.5, size=0.5, color=gray, opacity=0.3,style={shading=ball, ball color=gray}]{2}
\Vertex[x=0, y=-2, size=0.5, size=0.5,  color=gray,  opacity=0.3, style={shading=ball, ball color=gray}]{3}
\Vertex[x=0, y=-3, size=0.5, size=0.5,  color=gray,  opacity=0.3, style={shading=ball, ball color=gray}]{4}
\Vertex[x=5, y=.5, size=0.5, size=0.5,  color=gray,  opacity=0.3, style={shading=ball, ball color=gray}]{7} 
\Vertex[x=5, y=-.5, size=0.5, size=0.5,  color=gray,  opacity=0.3, style={shading=ball, ball color=gray}]{8} 
\Vertex[x=5, y=-2, size=0.5, size=0.5,  color=gray,  opacity=0.3, style={shading=ball, ball color=gray}]{9}
\Vertex[x=5, y=-3, size=0.5, size=0.5,  color=gray,  opacity=0.3, style={shading=ball, ball color=gray}]{10}
\Edge[Direct, lw=.7](S1)(1)
\Edge[Direct, lw=.7](S1)(2)
\Edge[Direct, lw=.7](S2)(3)
\Edge[Direct, lw=.7](S2)(4)
\Edge[Direct, lw=.7](7)(T1)
\Edge[Direct, lw=.7](8)(T1)
\Edge[Direct, lw=.7](9)(T2)
\Edge[Direct, lw=.7](10)(T2)
\Edge[Direct, lw=.7](1)(7)
\Edge[Direct, lw=.7](2)(8)
\Edge[Direct, lw=.7](3)(9)
\Edge[Direct, lw=.7](4)(10)
\Edge[Direct, lw=.7](3)(8)
\Edge[Direct, lw=.7](2)(9)
\end{tikzpicture} 
\caption{Network for Example~\ref{eee1}.}
\label{fig:next}
\end{figure}
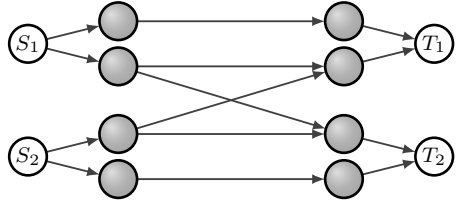 
\noindent By applying Theorem \ref{upperBound}, we obtain that for all $m \ge 1$ and all $(\alpha_1,\alpha_2) \in \calR_m(\mL)$ we have $\alpha_1 + \alpha_2 \le 3$, since
\begin{align*}
    \rank\begin{pmatrix} F_{11}\\ F_{21}\end{pmatrix} - \rank F_{21} +s_2
&= 2 - 1 + 2 = 3,\\
\rank\begin{pmatrix} F_{12}\\ F_{22}\end{pmatrix} - \rank F_{12} +s_1
&= 2 - 1 + 2
= 3.
\end{align*}  

For $m=1$, the $1$-codes $C = \left(\F_q^2, \gen{(0, 1)}\right)$ and $C=\left(\gen{(1,0)}, \F_q^2\right)$ are both unambiguous, meaning that the rates $(2,1)$ and $(1,2)$ are achievable in one shot. This implies that the upper bound of Theorem~\ref{upperBound} is tight in this case.  
\end{example}

\begin{example}\label{eee2}
Let $q$ be arbitrary and consider the $q$-LMUC 
$\mL=(2,(1,2), (1,2),F)$,
where
\[F = \begin{pmatrix} 1 & 1 & 1\\ 1 & 1 & 0\\ 1  & 0 & 1\end{pmatrix}. \]
Note that $\mL$ is induced, for example, by the network in Figure~\ref{fig:gennet}.
\begin{figure}[H]
\centering
\begin{tikzpicture}[scale=0.6]
\Vertex[label=$S_1$, x=-2, y=0, size=0.5, color=gray, opacity=0]{S1}
\Vertex[label=$S_2$, x=-2, y=-2.5, size=0.5, color=gray, opacity=0]{S2}
\Vertex[label=$T_1$, x=7, y=0, size=0.5, color=gray, opacity=0]{T1}
\Vertex[label=$T_2$, x=7, y=-2.5, size=0.5, color=gray, opacity=0]{T2}
\Vertex[size=0.5, size=0.5, color=gray, opacity=0.3,style={shading=ball, ball color=gray}]{1}
\Vertex[x=0, y=-1.5, size=0.5, size=0.5,  color=gray,  opacity=0.3, style={shading=ball, ball color=gray}]{2}
\Vertex[x=0, y=-3.5, size=0.5, size=0.5,  color=gray,  opacity=0.3, style={shading=ball, ball color=gray}]{3}
\Vertex[x=5, y=0, size=0.5, size=0.5,  color=gray,  opacity=0.3, style={shading=ball, ball color=gray}]{7}
\Vertex[x=5, y=-1.5, size=0.5, size=0.5,  color=gray,  opacity=0.3, style={shading=ball, ball color=gray}]{8}
\Vertex[x=5, y=-3.5, size=0.5, size=0.5,  color=gray,  opacity=0.3, style={shading=ball, ball color=gray}]{9}

\Edge[Direct, lw=.7](S1)(1)
\Edge[Direct, lw=.7](S2)(2)
\Edge[Direct, lw=.7](S2)(3)
\Edge[Direct, lw=.7](7)(T1)
\Edge[Direct, lw=.7](8)(T2)
\Edge[Direct, lw=.7](9)(T2)
\Edge[Direct, lw=.7](1)(7)
\Edge[Direct, lw=.7](2)(8)
\Edge[Direct, lw=.7](3)(9)
\Edge[Direct, lw=.7](1)(8)
\Edge[Direct, lw=.7](1)(9)
\Edge[Direct, lw=.7](2)(7)
\Edge[Direct, lw=.7](3)(7)
\end{tikzpicture} 
\caption{Network for Example~\ref{eee2}.}
\label{fig:gennet}
\end{figure}
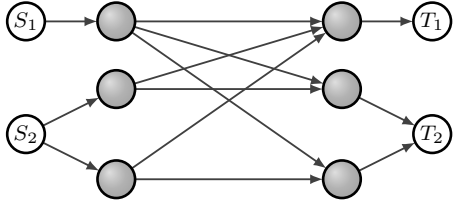 

By applying Theorem \ref{upperBound} we obtain that
for all $m \ge 1$ and for all $(\alpha_1,\alpha_2) \in \calR_m(\mL)$ we have $\alpha_1+\alpha_2\leq 2$.
Indeed, 
\begin{align*}
    \rank\begin{pmatrix} F_{11}\\ F_{21}\end{pmatrix} - \rank F_{21} +s_2
&= 1 - 1 + 2 = 2,  \\
\rank\begin{pmatrix} F_{12}\\ F_{22}\end{pmatrix} - \rank F_{12} +s_1
&= 2 - 1 + 1
= 2.
\end{align*}

In the next section, we will show that whether or not this bound is sharp depends 
on the \textit{characteristic} of the finite field $\F_q$ and not, for example, on $m$.
\end{example}

\section{The Role of the Field Characteristic}
\label{s:char}

The goal of this section is to illustrate the role that the field characteristic plays in the problem we are considering. Note that the problem we are studying in this section is an extension of the problem studied until now. When defining a $q$-LMUC, since $q$ is a given parameter of the network, the characteristic of the field is, in general, fixed. The problem we consider in this section is based on the remark that any field contains the neutral elements for addition and multiplication and that in general those are denoted by 0 and 1. This implies that any matrix with entries only in $\{0,1\}$ can be the transfer matrix of a $q$-LMUC for any prime power $q$. It is natural to look into the achievable rates regions across different fields for these types of $q$-LMUCs.

It is well known that the size of the field plays an important role in achieving the network capacity for multicast networks when the network code design is part of the problem, see \cite{LYC03,KM03}. We show that in sharp contrast with this scenario, the \textit{characteristic} of the underlying field plays an important role in our model. Note that it is not the first time the characteristic of the field has played an important role in network coding. In \cite{DFZ05}, the authors show networks for which capacities are achievable in either even or odd characteristic fields. It is worth repeating that our problem differs from the one studied in \cite{DFZ05} since the authors focus on constructing a network code, whereas in our case, the network code is frozen.  

We show that given the $q$-LMUC from Example \ref{eee2}, the achievable rates region of the network over an odd characteristic field strictly contains the achievable rates region  over an even characteristic field.

\begin{theorem}
Let $\mL=(2,(1,2),(1,2),F)$
be the $q$-LMUC  from Example \ref{eee2}. Recall that
$$F=
\begin{pmatrix}
1 & 1 & 1 \\ 1 & 1 & 0 \\ 1 & 0 & 1
\end{pmatrix}.$$
If $q$ is odd, then 
$(1,1) \in \calR_1(\mL) \subseteq \calR(\mL)$.
If $q$ is even, then for any $m \ge 1$ and any $(\alpha_1,\alpha_2) \in \calR_m(\mL)$ we have
\begin{equation}\label{e:final}
    2\alpha_1+\alpha_2 \le 2.
\end{equation}
In particular, $(1,1) \notin \calR(\mL)$.
\end{theorem}

\begin{remark}
Note that this result not only proves that the achievability regions of the network when using even or odd characteristics are different, but it also shows that the achievability region when communicating using an even characteristic field is strictly contained in the one obtained using an odd characteristic field. More specifically, an interested reader will be able to see that, when using even characteristic fields, the rate $(1,1)$ not only is not achievable but is bounded away, meaning that it cannot be achieved even with infinitely many uses of the network.
\end{remark}

\begin{proof}
Observe that if $q$
is odd, then the 1-code $C=\F_q\times\langle (1,-1)\rangle$ is unambiguous and therefore $(1,1) \in \calR_1(\mL)$, as claimed. 

To prove the second part of the theorem,
denote by
$f_{i,j}$ the multiplication by $F_{i,j}$ on the right,
for $i,j\in\{1,2\}$.
Let $(\alpha_1,\alpha_2) \in \calR_m(\mL)$
and let 
$C=C_1\times C_2$ be an unambiguous $m$-code 
with $\log_{q^m}(C_1)=\alpha_1$ and $\log_{q^m}(C_2)=\alpha_2$.
Recall that by Equation \eqref{e:fan_is},
we have that $|\Fan_1(x,C)|=|\IS_1(C)|$ for all $x\in C_1$. 
Since $C$ is unambiguous, we have $|C_1|\cdot |\IS_1(C)|\leq |\F_{q^m}|=q^m$, or equivalently 
\begin{equation}\label{e:is1}|\IS_1(C)|\leq \frac{q^{m}}{|C_1|}.\end{equation}

Recall that, by definition,
$\IS_1(C)=f_{2,1}(C_2)$.
Observe that for all
$x\in \F_{q^m}$
we have 
$f_{2,1}^{-1}(x)=\{(y,x-y)\mid y\in \F_{q^m}\}$  and therefore it holds that
\begin{equation*}\label{e:even-preimage}
    C_2\subseteq f_{2,1}^{-1}(f_{2,1}(C_2))=f_{2,1}^{-1}(\IS_1(C))=\bigcup_{x\in \IS_1(C)}f_{2,1}^{-1}(x)=\bigcup_{x\in \IS_1(C)}\{(y,x-y)\mid y\in \F_{q^m}\}.
\end{equation*}
In particular, all the elements of $C_2$ are of the form $(y,x-y)$ for some $x \in \IS_1(C)$ and $y \in \F_{q^m}$.
Now fix any $y\in \F_{q^m}$ and $x\in \IS_1(C)$ with $(y,x-y)\in C_2$, and observe that 
\begin{equation}\label{e:even-fan}
\Fan_2((y,x-y),C)=f_{2,2}(y,x-y)+\IS_2(C)=(y,x-y)+\langle(1,1)\rangle.
\end{equation}
Here is where the characteristic of the field starts to play a crucial role. If $q$ is even, then Equation \eqref{e:even-fan} can be rewritten as
\begin{align*}
    \Fan_2((y,x+y),C)= (y,x+y)+\langle(1,1)\rangle=\{(z,x+z)\mid x\in \F_{q^m}\}=f_{2,1}^{-1}(x).
\end{align*}
It follows that
\begin{equation} \label{eql}
    \bigcup_{(y,x+y)\in C_2}\Fan_2((y,x+y),C)\subseteq\bigcup_{x\in \IS_1(C)} f_{2,1}^{-1}(x)=f_{2,1}^{-1}(\IS_1(C)).
\end{equation}
Again by Equation \eqref{e:fan_is} and that $f_{1,2}$ is injective, we have 
$|\Fan_2((y,x+y),C)|=|\IS_2(C)|=|C_1|$. Combining this fact with Equation~\eqref{eql} and with the unambiguity of $C$ we obtain
\[|C_2|\leq \frac{|f_{2,1}^{-1}(\IS_1(C))|}{|\Fan_2((y,x+y),C)|}=\frac{|\IS_1(C)|\cdot |\F_{q^m}|}{|C_1|}\leq\left(\frac{q^m}{|C_1|}\right)^2,\]
where the last equality is a  consequence of Equation \eqref{e:is1}.
Taking the logarithm with base~$q^m$ of the inequality above, one gets
$\alpha_2 \le 2(1-\alpha_1)$, as desired.
\end{proof}


The following proposition shows that some of the rates satisfying Equation \eqref{e:final} are achievable.

\begin{proposition}
Let $\mL=(2,(1,2),(1,2),\F_{2^m})$
be the $2^m$-LMUC  from Example \ref{eee2}. 
Then for any $n\leq m$ we have  $\left(\frac{n}{m}, 2\left(1-\frac{n}{m}\right)\right) \in \calR(\mL)$.
\end{proposition}

\begin{proof}
 Let $\{x_1,\dots,x_m\}$ be an ordered basis of $\F_{2^m}$ over $\F_2$. Define $C_1=\langle x_1,\dots,x_n\rangle_{\F_2}$ and $C_2=\langle (x_i,0),(0,x_i)\mid i=n+1,\dots,m\rangle_{\F_2}$. We can compute the interference sets $\IS_1(C)=f_{2,1}(C_2)=\langle x_{n+1},\dots,x_m\rangle_{\F_2}$ and $\IS_2(C)=f_{1,2}(C_1)=\langle (x_1,x_1),\dots,(x_n,x_n)\rangle_{\F_2}$. Thus $C_1\times C_2$ is unambiguous because $C_1\cap \IS_1(C)=C_2\cap \IS_2(C)=\{0\}$.
\end{proof}

\section{Conclusions}
We considered the multiple unicast problem over a coded network where the network code is linear over a finite field and fixed. We introduced a  framework to define and investigate the achievable rate region in this context. We established an outer bound on the achievable rate region and provided examples where the outer bound is sharp. Finally, we illustrated the role played by the field characteristic in this problem, which is different than what is generally observed in the context of linear network coding.

\bibliographystyle{IEEEtran}
\bibliography{refs}

\end{document}